\documentclass[draft,a4paper,11pt]{article}
\usepackage[margin=1in]{geometry}

\usepackage{amsmath,amssymb,amsthm,multirow}
\usepackage[sort&compress,square,sort,comma,numbers]{natbib}
\usepackage[inline]{enumitem}
\usepackage[final]{microtype}
\usepackage[hyphens]{url}
\usepackage[hidelinks]{hyperref}

\newtheorem{theorem}{Theorem}[section]

\newtheorem{proposition}[theorem]{Proposition}

\theoremstyle{definition}

\newtheorem{definition}[theorem]{Definition}
\theoremstyle{remark}
\newtheorem{example}[theorem]{Example}

\newcommand{\Cluster}{\mathcal{C}}
\newcommand{\System}{\mathfrak{S}}
\newcommand{\NonFaulty}[1]{\mathop{\textsf{nf}}(#1)}
\newcommand{\Faulty}[1]{\mathop{\textsf{f}}(#1)}
\newcommand{\n}[1]{\mathbf{n}_{#1}}
\newcommand{\f}[1]{\mathbf{f}_{#1}}
\newcommand{\nf}[1]{\mathbf{nf}_{#1}}
\newcommand{\Replica}[1][r]{\textsc{#1}}

\newcommand{\Cert}[2]{\langle #1 \rangle_{#2}}
\newcommand{\Partition}[1]{\mathop{\textsf{partition}}(#1)}

\newcommand{\abs}[1]{\lvert #1 \rvert}
\newcommand{\Size}[1]{\lVert #1 \rVert}

\newcommand{\intersect}{\cap}
\newcommand{\difference}{\setminus}
\newcommand{\BigO}[1]{\mathcal{O}(#1)}

\renewcommand{\div}{\operatorname{div}}
\renewcommand{\bmod}{\operatorname{mod}}
\newcommand{\sgn}{\operatorname{sgn}}

\usepackage[noend]{algorithmic}
\newcommand{\Name}[1]{\textnormal{\texttt{#1}}}

\newenvironment{myprotocol}{
    \hrule
    \smallskip
    \small
    \begin{algorithmic}[1]
        \newcommand{\SPACE}{\item[]}
        \newcommand{\TITLE}[1]{\item[] \textbf{\underline{##1}:}\\[2pt]}
        \makeatletter
            \newcommand{\EVENT}[1]{\STATE \textbf{event} ##1 \textbf{do}\begin{ALC@g}}
            \newcommand{\ENDEVENT}{\end{ALC@g}}
        \makeatother
}{
    \end{algorithmic}
    \smallskip
    \hrule
}

\usepackage{tikz}
\usetikzlibrary{arrows.meta,decorations.pathmorphing,decorations.pathreplacing}
\tikzset{
    >=Stealth,
    dot/.style={circle,scale=0.35,draw=black,fill=black},
    node_text/.append style={font=\strut\bfseries},
    label/.append style={font=\strut\footnotesize},
    decosnake/.append style={decoration={snake,pre length=4pt,post length=6pt,segment length=4,amplitude=.9},decorate}
}

\begin{document}

\title{The fault-tolerant cluster-sending problem\footnote{A brief announcement of this work will be presented at the 33rd International Symposium on Distributed Computing (DISC 2019)~\cite{clussendba}.}}
\author{Jelle Hellings \and Mohammad Sadoghi}
\date{\normalsize{\begin{tabular}{c}
                            Exploratory Systems Lab\\
                            Department of Computer Science\\
                            University of California, Davis\\
                            CA 95616-8562, USA
                \end{tabular}}}

\maketitle

\begin{abstract}
The development of fault-tolerant distributed systems that can tolerate Byzantine behavior has traditionally been focused on consensus protocols, which support fully-replicated designs. For the development of more sophisticated high-performance Byzantine distributed systems, more specialized fault-tolerant communication primitives are necessary, however.

In this paper, we identify an essential communication primitive and study it in depth. In specifics, we formalize the \emph{cluster-sending problem}, the problem of sending a message from one Byzantine cluster to another Byzantine cluster in a reliable manner. We not only formalize this fundamental problem, but also establish lower bounds on the complexity of this problem under crash failures and Byzantine failures. Furthermore, we develop practical cluster-sending protocols that meet these lower bounds and, hence, have optimal complexity. As such, our work provides a strong foundation for the further exploration of novel designs that address challenges encountered in fault-tolerant distributed systems.
\end{abstract}

\section{Introduction}
Recently, the emergence of \emph{blockchain technology} has fueled a renewed interest in the development of fault-tolerant distributed systems in which some of the participating replicas behave malicious~\cite{pwcenergy,christies,impactblock,wurblockfood,hypereal,blockdev,blockplane,encybd,blockhealthover,blockhealthfac,blockeu,promiseblock,ibmgdpr,bitcoin,ethereum,mbftba}. The main focus of current developments is mostly limited to fully-replicated systems in which each participating replica has the same role. The benefit of such a fully-replicated design is that one can rely on readily-available \emph{consensus protocols} to implement such a design~\cite{paxos,pbft,distalgo}.

We envision the design and development of more sophisticated high-performance Byzantine systems in which replicas have specialized roles. An example of such a system would be a \emph{sharded geo-scale design} in which data is kept in \emph{local Byzantine clusters}. In such a sharded geo-scale design, many queries can efficiently be answered by involving only a single cluster. In this way, a sharded design will often improve scalability and performance when dealing with massive large-scale databases~\cite{distsys,dbdist}. For answering more complex queries, we need cooperation between different clusters, however.

Hence, to enable the design and development of such systems, we need reliable ways for Byzantine clusters to communicate and cooperate. We believe that the existing consensus protocols are insufficient to fulfill this aim~\cite{pbft,pbftj, zyzzyva,zyzzyvaj,bft700,rbft,steward,algorand,aardvark,scaling,cheapbft,minbft,fastbft,distalgo}: we can run a single global consensus protocol among all replicas in all clusters to enable sharing of data and queries, but this would be at high---\emph{quadratic}---communication costs for all replicas involved and would eliminate any possible scaling benefits of a clustered design. Indeed, we believe that there is a pressing need for more specialized Byzantine communication primitives. In this paper we formalize one such primitive, the \emph{cluster-sending problem}: the problem of sending a message from one Byzantine cluster to another Byzantine cluster in a reliable manner that is verifiable by all replicas involved. Our main contributions are as follows:
\begin{enumerate}[label=(\arabic*)]
    \item We formalize the cluster-sending problem.
    \item We prove strict lower bounds on the complexity of the cluster-sending problem in terms of the number of messages (when faulty replicas only crash) and the number of signatures (when faulty replicas can be malicious and messages are signed). In both cases, these lower bounds are only \emph{linear} in the size of the clusters involved.
    \item We introduce \emph{bijective sending}, a powerful technique to reliably send messages between clusters of roughly the same size. To generalize bijective sending to arbitrary-sized clusters, we introduce \emph{partitioned bijective sending} techniques.
    \item For many practical environments, we develop optimal cluster-sending protocols that use (partitioned) bijective sending and whose complexity matches the lower bounds established. A full overview of all the environmental conditions we study and corresponding protocols we propose can be found in Figure~\ref{fig:overview}.
\end{enumerate}

\begin{figure}
    \centering
    \makebox[0pt]{
        \begin{tabular}{l|| lllll}
            &Protocol&System&Robustness&Messages&Message size\\
        \hline
        \hline
           \multirow{4}{*}{\rotatebox[origin=c]{90}{non-linear}}
            &\Name{RB-bcs}&Omit   &$\n{\Cluster_1} > 2\f{\Cluster_1}$, $\n{\Cluster_2} > \f{\Cluster_2}$&$(\f{\Cluster_1} + 1)\cdot(\f{\Cluster_2} + 1)$&$\BigO{\Size{v}}$\\
            &\Name{RB-brs}&Byzantine, RS&$\n{\Cluster_1} > 2\f{\Cluster_1}$, $\n{\Cluster_2} > \f{\Cluster_2}$&$(2\f{\Cluster_1} + 1)\cdot(\f{\Cluster_2} + 1)$&$\BigO{\Size{v}}$\\
            &\Name{RB-bcs}&Byzantine, RS&$\n{\Cluster_1} > 2\f{\Cluster_1}$, $\n{\Cluster_2} > \f{\Cluster_2}$&$(\f{\Cluster_1} + 1)\cdot(\f{\Cluster_2} + 1)$&$\BigO{\Size{v} + \f{\Cluster_1}}$\\
            &\Name{RB-bcs}&Byzantine, CS&$\n{\Cluster_1} > 2\f{\Cluster_1}$, $\n{\Cluster_2} > \f{\Cluster_2}$&$(\f{\Cluster_1} + 1)\cdot(\f{\Cluster_2} + 1)$&$\BigO{\Size{v}}$\\
        \hline
           \multirow{4}{*}{\rotatebox[origin=c]{90}{linear}}
            &\Name{PBS-bcs}  &Omit &$\n{\Cluster_1} > 3\f{\Cluster_1}$, $\n{\Cluster_2} > 3\f{\Cluster_2}$&$\BigO{\max(\n{\Cluster_1}, \n{\Cluster_2})}$  (optimal)&$\BigO{\Size{v}}$\\
            &\Name{PBS-brs}&Byzantine, RS&$\n{\Cluster_1} > 4\f{\Cluster_1}$, $\n{\Cluster_2} > 4\f{\Cluster_2}$&$\BigO{\max(\n{\Cluster_1}, \n{\Cluster_2})}$  (optimal)&$\BigO{\Size{v}}$\\
            &\Name{PBS-bcs}&Byzantine, RS&$\n{\Cluster_1} > 3\f{\Cluster_1}$, $\n{\Cluster_2} > 3\f{\Cluster_2}$&$\BigO{\max(\n{\Cluster_1}, \n{\Cluster_2})}$  &$\BigO{\Size{v}+ \f{\Cluster_1}}$\\
            &\Name{PBS-bcs}&Byzantine, CS&$\n{\Cluster_1} > 3\f{\Cluster_1}$, $\n{\Cluster_2} > 3\f{\Cluster_2}$&$\BigO{\max(\n{\Cluster_1}, \n{\Cluster_2})}$  (optimal)&$\BigO{\Size{v}}$
        \end{tabular}
    }
    \caption{Overview of cluster-sending protocols that sends a value $v$ from cluster $\Cluster_1$ to cluster $\Cluster_2$. In the above, RS is shorthand for \emph{replica signing}, CS is a shorthand for \emph{cluster signing}, and \Name{PBS} is a shorthand for the relevant instances of \Name{BS}, \Name{SPBS}, and \Name{RPBS}, which are protocols that use (partitioned) bijective sending.}\label{fig:overview}
\end{figure}

\paragraph{Organization}
In Section~\ref{sec:formal}, we introduce the terminology used throughout this paper and formally define the \emph{cluster-sending problem}. In Section~\ref{sec:rb}, we show how to use \emph{reliable broadcasting} as straightforward basic technique to solve the cluster-sending problem in all possible settings. Next, in Section~\ref{sec:lower}, we prove \emph{lower bounds} on the complexity of the cluster-sending problem. Then, in Section~\ref{sec:bijective}, we introduce \emph{bijective sending}, a powerful cluster-sending technique that performs cluster-sending with minimal communication between clusters of comparable sizes in which a minority of all replicas are faulty. Next, in Section~\ref{sec:partitioned}, we introduce \emph{partition techniques} that allow for the generalization of bijective sending to clusters of arbitrary sizes. Finally, in Section~\ref{sec:conclude}, we conclude on our findings and discuss avenues for future work.

\section{Formalizing the cluster-sending problem}\label{sec:formal}

A \emph{cluster} $\Cluster$ is a set of replicas. We write $\Faulty{\Cluster} \subseteq \Cluster$ to denote the set of \emph{faulty replicas} in $\Cluster$ and $\NonFaulty{\Cluster} = \Cluster \difference \Faulty{\Cluster}$ to denote the set of \emph{non-faulty replicas} in $\Cluster$. We write $\n{\Cluster} = \abs{\Cluster}$, $\f{\Cluster} = \abs{\Faulty{\Cluster}}$, and  $\nf{\Cluster} = \abs{\NonFaulty{\Cluster}}$ to denote the number of replicas, faulty replicas, and non-faulty replicas in the cluster, respectively.  We extend the notations $\Faulty{\cdot}$, $\NonFaulty{\cdot}$, $\n{(\cdot)}$, $\f{(\cdot)}$, and $\nf{(\cdot)}$ to arbitrary sets of replicas. In this work, we consider faulty replicas that can \emph{crash}, \emph{omit} messages, or behave \emph{Byzantine}. A \emph{crashing replica} executes steps correctly up till some point, after which it does not execute anything. An \emph{omitting replica} executes steps correctly, but can decide to not send a message when it should or decide to ignore messages it receives. A \emph{Byzantine replica} can behave in arbitrary, possibly coordinated and malicious, manners.

A \emph{cluster system} $\System$ is a finite set of clusters such that communication between replicas in a cluster is \emph{local} and communication between clusters is \emph{non-local}. We assume that there is no practical bound on local communication (e.g., within a single data center rack), while global communication is limited, costly, and to be avoided (e.g., between data centers in different continents). If $\Cluster_1, \Cluster_2 \in \System$ are distinct clusters, then we assume that $\Cluster_1 \intersect \Cluster_2 = \emptyset$: no replica is part of two distinct clusters.

\begin{definition}\label{def:csp}
Let $\System$ be a system and $\Cluster_1, \Cluster_2 \in \System$ be two clusters with non-faulty replicas ($\NonFaulty{\Cluster_1} \neq \emptyset$ and $\NonFaulty{\Cluster_2} \neq \emptyset$). The \emph{cluster-sending problem} is the problem of sending a value $v$ from $\Cluster_1$ to $\Cluster_2$ such that:
\begin{enumerate}
\item all non-faulty replicas in $\Cluster_2$ \emph{receive} the value $v$;
\item only if all non-faulty replicas in $\Cluster_1$ \emph{agree} upon sending the value $v$ to $\Cluster_2$ will non-faulty replicas in $\Cluster_2$ receive $v$; and
\item all non-faulty replicas in $\Cluster_1$ can \emph{confirm} that the value $v$ was received.
\end{enumerate}
\end{definition}

In the following, we assume \emph{asynchronous reliable communication}: all messages send by non-faulty replicas eventually arrive at their destination. None of the protocols we propose rely on message delivery timings for their correctness.  Let $\Cluster \in \System$ be a cluster and $\Replica \in \Cluster$ be a replica. We assume that, on receipt of a message $m$ from replica $\Replica$, one can determine that $\Replica$ did sent $m$ if $\Replica \in \NonFaulty{\Cluster}$; and one can only determine that $m$ was sent by a non-faulty replica if $\Replica \in \NonFaulty{\Cluster}$. Hence, faulty replicas are able to impersonate each other, but are not able to impersonate non-faulty replicas. We study the \emph{cluster-sending problem} for Byzantine systems in two different types of environments:
\begin{enumerate}
\item A system provides \emph{replica signing} if every replica $\Replica$ can \emph{sign} arbitrary messages $m$, resulting in a certificate $\Cert{m}{\Replica}$. These certificates are non-forgeable and can be constructed only if $\Replica$ cooperates in constructing them. Based on only the certificate $\Cert{m}{\Replica}$, anyone can verify that $m$ was originally supported by $\Replica$ (unless $\Replica \in \Faulty{\Cluster}$). 
\item A system provides \emph{cluster signing} if it is equipped with a \emph{signature scheme} that can be used to \emph{cluster-sign} arbitrary messages $m$, resulting in a certificate $\Cert{m}{\Cluster}$. These certificates are non-forgeable and can be constructed only if all non-faulty replicas in $\NonFaulty{\Cluster}$ cooperate in constructing them. Based on only the certificate $\Cert{m}{\Cluster}$, anyone can verify that $m$ was originally supported by all non-faulty replicas in $\Cluster$.
\end{enumerate}

In practice, \emph{replica signing} can be implemented using \emph{digital signatures}, which rely on a public-key cryptography infrastructure~\cite{hac}, and \emph{cluster signing} can be implemented using threshold signatures, which are available for some public-key cryptography infrastructures~\cite{rsasign}. Let $m$ be a message, $\Cluster \in \System$ a cluster, and $\Replica \in \Cluster$ a replica. We write $\Size{v}$ to denote the size of any arbitrary value $v$. We assume that the size of certificates $\Cert{m}{\Replica}$, obtained via replica signing, and certificates $\Cert{m}{\Cluster}$, obtained via cluster signing, are both linearly upper-bounded by $\Size{m}$. More specifically, $\Size{(m, \Cert{m}{\Replica})} = \BigO{\Size{m}}$ and $\Size{(m, \Cert{m}{\Cluster})} = \BigO{\Size{m}}$.

We notice that \emph{cluster signing} can be \emph{emulated using replica signing}. If $\Cluster \in \System$ is a cluster, $m$ is a message, and non-faulty replicas in $\NonFaulty{\Cluster}$ only provide certificates $\Cert{m}{\Replica}$ if there is consensus on doing so among all non-faulty replicas in $\NonFaulty{\Cluster}$, then the set $\{ \Cert{m}{\Replica} \mid \Replica \in S \}$, for any set $S \subseteq \Cluster$ with $\abs{S} = \f{\Cluster} + 1$, can be used in the same manner as a cluster certificate $\Cert{m}{\Cluster}$. In this case, we have $\Size{(m, \Cert{m}{\Cluster})} =  \BigO{\Size{m} + \f{\Cluster}}$, however. If we assume only crash or omission failures, then no replica will ever try to forge messages of other replicas or send messages outside the scope of the relevant protocol. Hence, in this setting, \emph{replica signing} or \emph{cluster signing} does not add any reliability, implying there is no need for certificates. In this case, we simply emulate replica or cluster certificates by not including them.

When necessary, we assume that replicas in each cluster $\Cluster \in \System$ can reach agreement on a value using an off-the-shelf \emph{consensus protocol}~\cite{pbft,pbftj, zyzzyva,zyzzyvaj,bft700,rbft,steward,algorand,aardvark,scaling,cheapbft,minbft,fastbft,distalgo}. In the best case, when we only have crash failures or when we have synchronous communication and replica signing (within a cluster), these protocols require $\n{\Cluster} > 2\f{\Cluster}$, which we assume to be the case for all \emph{sending clusters}.\footnote{Strictly speaking there exist synchronous authenticated consensus protocols that can reach agreement on a value among all non-faulty replicas even if $\n{\Cluster} \leq 2\f{\Cluster}$, e.g.~\cite{consbound,dolevstrong}. Unfortunately, an outside observer---including other clusters---will never be able to reliable learn this value, as it will never be able to distinguish between the faulty and the non-faulty replicas.}

In this paper, we use the notation $i\sgn{j}$, with $i, j \geq 0$ and $\sgn$ the sign function, to denote $i$ if $j > 0$ and $0$ otherwise.

\section{Cluster-sending via reliable broadcasts}\label{sec:rb}
A principle technique used by consensus protocols to guarantee agreement of non-faulty replicas is \emph{message broadcasting} (e.g., as used in \Name{Paxos} and \Name{Pbft}~\cite{paxos,paxossimple,pbft,pbftj}). We can use message broadcasting in the construction of simple cluster-sending protocols, which can be used as a baseline for comparisons. 

First, we present a broadcast-based protocol that can operate in a system with Byzantine failures and cluster certificates. In this protocol, cluster $\Cluster_1$ uses a consensus protocol to reach agreement on a value $v$. Then, a set $S_1 \subseteq \Cluster_1$ of $\f{\Cluster_1}+1$ replicas in $\Cluster_1$ and a set $S_2 \subseteq \Cluster_2$ of $\f{\Cluster_2}+1$ replicas in $\Cluster_2$ are chosen. Finally, each replica in $S_1$ is instructed to broadcast $v$ to all replicas in $S_2$. Due to the size of $S_1$ and $S_2$, at least one non-faulty replica in $\Cluster_1$ will send a value to a non-faulty replica in $\Cluster_2$, which is sufficient to bootstrap \emph{receipt} and \emph{confirmation} of $v$ in $\Cluster_2$.  The pseudo-code for this protocol, named \Name{RB-bcs}, can be found in Figure~\ref{fig:rb_bcs}. Next, we prove the correctness of \Name{RB-bcs}:

\begin{figure}[t!]
    \begin{minipage}{0.7\textwidth}
        \begin{myprotocol}
            \TITLE{Protocol for the sending cluster $\Cluster_1$}
            \STATE Agree on $v$ and distribute $(v, \Cert{v}{\Cluster_1})$ to all non-faulty replicas in $\Cluster_1$.
            \STATE Choose replicas $S_1 \subseteq \Cluster_1$ with $\n{S_1} = \f{\Cluster_1} + 1$.
            \STATE Choose replicas $S_2 \subseteq \Cluster_2$ with $\n{S_2} = \f{\Cluster_2} + 1$.
            \FOR{$\Replica_1 \in S_1$}
                \FOR{$\Replica_2 \in S_2$}
                    \STATE $\Replica_1$ sends $(v, \Cert{v}{\Cluster_1})$ to $\Replica_2$.\label{fig:rb_bcs:bc}
                \ENDFOR
            \ENDFOR
            \SPACE
            \TITLE{Protocol for the receiving cluster $\Cluster_2$}
            \EVENT{$\Replica_2 \in \NonFaulty{\Cluster_2}$ receives $(w, \Cert{w}{\Cluster_1})$ from a replica in $\Cluster_1$}\label{fig:rb_bcs:cluster_receipt}
                \STATE Broadcast $(w, \Cert{w}{\Cluster_1})$ to all replicas in $\Cluster_2$.
            \ENDEVENT
            \EVENT{$\Replica_2' \in \NonFaulty{\Cluster_2}$ receives $(w, \Cert{w}{\Cluster_1})$ from a replica in $\Cluster_2$}\label{fig:rb_bcs:local_receipt}
                \STATE $\Replica_2'$ considers $w$ \emph{received}.
            \ENDEVENT
        \end{myprotocol}
    \end{minipage}
    \hfill
    \begin{minipage}{0.285\textwidth}
        \caption{\Name{RB-bcs}, the reliable broadcast cluster-sending protocol that sends a value $v$ from $\Cluster_1$ to $\Cluster_2$. We assume Byzantine failures and a system that provides cluster signing.}\label{fig:rb_bcs}
    \end{minipage}
\end{figure}

\begin{proposition}\label{prop:rb_bcs}
Let $\System$ be a system with Byzantine failures and cluster signing and let $\Cluster_1, \Cluster_2 \in \System$ . If $\n{\Cluster_1} > 2\f{\Cluster_1}$ and $\n{\Cluster_2} > \f{\Cluster_2}$, then \Name{RB-bcs} satisfies Definition~\ref{def:csp}. The protocol sends $(\f{\Cluster_1} + 1)\cdot(\f{\Cluster_2} + 1)$ messages, of size $\BigO{\Size{v}}$ each, between $\Cluster_1$ and $\Cluster_2$.
\end{proposition}
\begin{proof}
Choose $S_1 \subseteq \Cluster_1$ and $S_2 \subseteq \Cluster_2$ in accordance with \Name{RB-bcs}. We have $\n{S_1} = \f{\Cluster_1} + 1$ and $\n{S_2} = \f{\Cluster_2} + 1$. By construction, we have $\nf{S_1} \geq 1$ and $\nf{S_2} \geq 1$. Due to Line~\ref{fig:rb_bcs:bc}, each replica $\Replica_2 \in \NonFaulty{S_2}$ will receive the message $(v, \Cert{v}{\Cluster_1})$ from every replica in $\NonFaulty{S_1}$. As $\nf{S_1} \geq 1$, every $\Replica_2 \in \NonFaulty{S_2}$ will meet the condition at Line~\ref{fig:rb_bcs:cluster_receipt} and broadcast $(v, \Cert{v}{\Cluster_1})$ to all replicas in $\Cluster_2$. As $\nf{S_2} \geq 1$, each replica $\Replica_2' \in  \NonFaulty{\Cluster_2}$ will meet the condition at Line~\ref{fig:rb_bcs:local_receipt}, proving \emph{receipt} and \emph{confirmation}. We have \emph{agreement}, as $\Cert{v}{\Cluster_1}$ is non-forgeable.
\end{proof}

As replica signing can emulate cluster signing, \Name{RB-bcs} can also be used for systems with only replica signing. Such an emulated solution does require large messages whose size depends on the size of the sending cluster, however. For systems with replica signing we can improve on \Name{RB-bsv} in another manner. We propose \Name{RB-brs}, for which the pseudo-code can be found in Figure~\ref{fig:rb_brs}. Next, we prove the correctness of \Name{RB-brs}:

\begin{figure}[t!]
    \begin{minipage}{0.7\textwidth}
        \begin{myprotocol}
        \TITLE{Protocol for the sending cluster $\Cluster_1$}
        \STATE Agree on $v$ and distribute $v$ to all non-faulty replicas in $\Cluster_1$.
        \STATE Choose replicas $S_1 \subseteq \Cluster_1$ with $\n{S_1} = 2\f{\Cluster_1} + 1$.
        \STATE Choose replicas $S_2 \subseteq \Cluster_2$ with $\n{S_2} = \f{\Cluster_2} + 1$.
        \FOR{$\Replica_1 \in S_1$}
            \FOR{$\Replica_2 \in S_2$}
                \STATE $\Replica_1$ sends $(v, \Cert{v}{\Replica_1})$ to $\Replica_2$.\label{fig:rb_brs:bc}
            \ENDFOR
        \ENDFOR
        \SPACE
        \TITLE{Protocol for the receiving cluster $\Cluster_2$}
        \EVENT{$\Replica_2 \in \NonFaulty{\Cluster_2}$ receives $(w, \Cert{w}{\Replica_1'})$ from a replica $\Replica_1' \in \Cluster_1$}\label{fig:rb_brs:cluster_receipt}
            \STATE Broadcast $(w, \Cert{w}{\Replica_1'})$ to all replicas in $\Cluster_2$.
        \ENDEVENT
        \EVENT{$\Replica_2' \in \NonFaulty{\Cluster_2}$ receives $\f{\Cluster_1} + 1$ messages $(w, \Cert{w}{\Replica_1'})$ such that:
            \begin{enumerate}[nosep,label=(\alph*)]
                \item each message is sent by a replica in $\Cluster_2$; and
                \item each message includes a $\Cert{w}{\Replica_1'}$ from distinct replicas $\Replica_1' \in \Cluster_1$
            \end{enumerate}
        }\label{fig:rb_brs:local_receipt}
            \STATE $\Replica_2'$ considers $w$ \emph{received}.
        \ENDEVENT
        \end{myprotocol}
    \end{minipage}
    \hfill
    \begin{minipage}{0.285\textwidth}
        \caption{\Name{RB-brs}, the reliable broadcast cluster-sending protocol that sends a value $v$ from $\Cluster_1$ to $\Cluster_2$. We assume Byzantine failures and a system that provides replica signing.}\label{fig:rb_brs}
    \end{minipage}
\end{figure}

\begin{proposition}\label{prop:rb_brs}
Let $\System$ be a system with Byzantine failures and replica signing and let $\Cluster_1, \Cluster_2 \in \System$. If $\n{\Cluster_1} > 2\f{\Cluster_1}$ and $\n{\Cluster_2}> \f{\Cluster_2}$, then \Name{RB-brs} satisfies Definition~\ref{def:csp}. The protocol sends $(2\f{\Cluster_1} + 1)\cdot(\f{\Cluster_2} + 1)$ messages, of size $\BigO{\Size{v}}$ each, between $\Cluster_1$ and $\Cluster_2$.
\end{proposition}
\begin{proof}
Choose $S_1 \subseteq \Cluster_1$ and $S_2 \subseteq \Cluster_2$ in accordance with \Name{RB-brs}. We have $\n{S_1} = 2\f{\Cluster_1} + 1$ and $\n{S_2} = \f{\Cluster_2} + 1$. By construction, we have $\nf{S_1} \geq \Faulty{\Cluster_1} + 1$ and $\nf{S_2} \geq 1$. Due to Line~\ref{fig:rb_brs:bc}, each replica $\Replica_2 \in \NonFaulty{S_2}$ will receive messages $(v, \Cert{v}{\Replica_1})$ from every replica in $\Replica_1 \in \NonFaulty{T_1}$. Hence, $\Replica_2$ will meet the condition at Line~\ref{fig:rb_brs:cluster_receipt} for each such message $(v, \Cert{v}{\Replica_1})$ and broadcast these messages to all replicas in $\Cluster_2$. As $\nf{S_2} \geq 1$ and $\nf{S_1} \geq \f{\Cluster_1} + 1$, each replica $\Replica_2' \in \NonFaulty{\Cluster_2}$ will meet the condition at Line~\ref{fig:rb_brs:local_receipt}, proving \emph{receipt} and \emph{confirmation}.

To prove \emph{agreement}, we show that only values agreed upon by $\Cluster_1$ will be considered received by non-faulty replicas in $\NonFaulty{\Cluster_2}$. Consider a value $v'$ not agreed upon by $\Cluster_1$. Hence, only the replicas in $\Faulty{\Cluster_1}$ will sign $v'$. Due to non-forgeability of replica certificates, the only certificates constructed for $v'$ are of the form $\Cert{v'}{\Replica_1}$, $\Replica_1 \in \Faulty{\Cluster_1}$. Consequently, each replica in $\Cluster_2$ can only receive and broadcast up to $\f{\Cluster_1}$ distinct messages of the form $(v', \Cert{v'}{\Replica_1'})$, $\Replica_1' \in \Cluster_1$. We conclude that no non-faulty replica will meet the conditions for $v'$ at Line~\ref{fig:rb_brs:local_receipt}.
\end{proof}

\section{Lower bounds for the cluster-sending problem}\label{sec:lower}

In the previous sections, we formalized the cluster-sending problem and considered broadcasting-based protocols to solve this problem. Unfortunately, these broadcasting-based protocols have high communication costs that, in the worst case, are quadratic in the size of the clusters involved. To determine whether we can do better than broadcasting, we will study the \emph{lower bound} on communication costs for any protocol solving the cluster-sending problem.

First, we consider systems with only crash failures, in which case we can lower bound the number of messages exchanged. This lower bound is entirely determined by the maximum number of messages that can get \emph{lost} due to faulty replicas not sending messages or ignoring received messages. In situations in which some replicas need to send or receive \emph{multiple} messages, the capabilities of faulty replicas to ignore messages is likewise multiplied. E.g., when the number of senders outnumbers the receivers, then some receivers must receive multiple messages. As these receivers could be faulty, this means they could cause loss of multiple messages. By a thorough analysis, we end up with the following lower bounds:

\begin{theorem}\label{thm:crash_opti}
Let $\System$ be a system with crash failures, let $\Cluster_1, \Cluster_2 \in \System$, and let $\{i, j\} = \{1,2\}$ such that $\n{\Cluster_i} \geq \n{\Cluster_j}$. Let $q_i = (\f{\Cluster_i} + 1) \div \nf{\Cluster_j}$, let $r_i = (\f{\Cluster_i} + 1) \bmod \nf{\Cluster_j}$, and let $\sigma_i = q_i\n{\Cluster_j} + r_i + \f{\Cluster_j}\sgn{r_i}$. Any protocol that solves the cluster-sending problem in which $\Cluster_1$ sends a value $v$ to $\Cluster_2$ needs to exchange at least $\sigma_i$ messages.
\end{theorem}
\begin{proof}
We assume $i = 1$, $j = 2$, and $\n{\Cluster_1} \geq \n{\Cluster_2}$. The proof is by contradiction. Hence, assume that a protocol \Name{P} can solve the cluster-sending problem using at most $\sigma_1 - 1$ messages. Consider a run of \Name{P} that sends messages $M$. Without loss of generality, we can assume that $\abs{M} = \sigma_1-1$. Let $R$ be the top $\f{\Cluster_2}$ receivers of messages in $M$, let $S = \Cluster_2 \difference R$, let $M_R \subset M$ be the messages received by replicas in $R$, and let $N = M \difference M_R$. We notice that $\n{R} = \f{\Cluster_2}$ and that $\n{S} = \nf{\Cluster_2}$. 

First, we prove that $\abs{M_R} \geq q_1\f{\Cluster_2} + \f{\Cluster_2}\sgn{r_1}$, this by contradiction. Assume $\abs{M_R} = q_1\f{\Cluster_2} + \f{\Cluster_2}\sgn{r_1} - v$, $v \geq 1$. Hence, we must have $\abs{N} = q_1\nf{\Cluster_2} + r_1 + v - 1$. Based on the value $r_1$, we distinguish two cases. The first case is $r_1 = 0$. In this case, $\abs{M_R} = q_1\f{\Cluster_2} - v < q_1\f{\Cluster_2}$ and $\abs{N} =  q_1\nf{\Cluster_2} + v - 1 \geq q_1\nf{\Cluster_2}$. As $q_1\f{\Cluster_2} > \abs{M_R}$, there must be a replica in $R$ that received at most $q_1-1$ messages. As $\abs{N} \geq q_1\nf{\Cluster_2}$, there must be a replica in $S$ that received at least $q_1$ messages. The other case is $r_1 > 0$. In this case, $\abs{M_R} = q_1\f{\Cluster_2} + \f{\Cluster_2} - v < (q_1+1)\f{\Cluster_2}$ and $\abs{N} =  q_1\nf{\Cluster_2} + r_1 + v - 1 > q_1\nf{\Cluster_2}$. As $(q_1+1)\f{\Cluster_2} > \abs{M_R}$, there must be a replica in $R$ that received at most $q_1$ messages. As $\abs{N} > q_1\nf{\Cluster_2}$, there must be a replica in $S$ that received at least $q_1+1$ messages. In both cases, we identified a replica in $S$ that received more messages than a replica in $R$, a contradiction. Hence, we must conclude that $\abs{M_R} \geq q_1\f{\Cluster_2} + \f{\Cluster_2}\sgn{r_1}$ and, consequently, $\abs{N} \leq q_1\nf{\Cluster_2} + r_1 - 1 \leq \f{\Cluster_1}$. As $\n{R} = \f{\Cluster_2}$, all replicas in $R$ could have crashed, in which case only the messages in $N$ are actually received. As $\abs{N} \leq \f{\Cluster_1}$, all messages in $N$ could be sent by replicas that have crashed. Hence, in the worst case, no message in $M$ is successfully sent by a non-faulty replica in $\Cluster_1$ and received by a non-faulty replica in $\Cluster_2$, implying that \Name{P} fails.
\end{proof}

Notice that the above lower bounds guarantee the delivery of at least one message. Next, we look at systems with Byzantine failures and replica signing. In this environment, we prove a lower bound on the number of certificates exchanged. In this case, the receiving cluster $\Cluster_2$ must eventually receive $\f{\Cluster_1} + 1$ distinct certificates signed by distinct replicas in $\Cluster_1$. A thorough analysis reveals the following lower bounds:
\begin{theorem}\label{thm:cert_opti}
Let $\System$ be a system with Byzantine failures and replica signing and let $\Cluster_1, \Cluster_2 \in \System$. Consider the cluster-sending problem in which $\Cluster_1$ sends a value $v$ to $\Cluster_2$.
\begin{enumerate}[nosep]
    \item Let $q_1 = (2\f{\Cluster_1} + 1) \div \nf{\Cluster_2}$, $r_1 = (2\f{\Cluster_1} + 1) \bmod \nf{\Cluster_2}$, and $\tau_1  = q_1\n{\Cluster_2} + r_1 + \f{\Cluster_2}\sgn{r_1}$. If $\n{\Cluster_1} \geq \n{\Cluster_2}$, then any protocol that solves the cluster-sending problem needs to exchange at least $\tau_1$ certificates of the form $\Cert{v}{\Replica}$, $\Replica \in \Cluster_1$.
    \item Let $q_2 = (\f{\Cluster_2} + 1) \div  (\nf{\Cluster_1} - \f{\Cluster_1})$, $r_2 = (\f{\Cluster_2} + 1) \bmod  (\nf{\Cluster_1} - \f{\Cluster_1})$, and $\tau_2  = q_2\n{\Cluster_1} + r_2 + 2\f{\Cluster_1}\sgn{r_2}$. If $\n{\Cluster_2} \geq \n{\Cluster_1}$, then any protocol that solves the cluster-sending problem needs to exchange at least $\tau_2$ certificates of the form $\Cert{v}{\Replica}$, $\Replica \in \Cluster_1$.
\end{enumerate}
\end{theorem}
\begin{proof}
For simplicity, we assume that each certificate is send to $\Cluster_2$ in an individual message independent of the other certificates. Hence, each certificate has a sender and a signer (both replicas in $\Cluster_1$) and a receiver (a replica in $\Cluster_2$). 

First, we prove the case for $\n{\Cluster_1} \geq \n{\Cluster_2}$ using contradiction. Assume that a protocol \Name{P} can solve the cluster-sending problem using at most $\tau_1 - 1$ certificates. Consider a run of \Name{P} that sends messages $C$, each message representing a single certificate, with $\abs{C} = \tau_1 - 1$. Following the proof of Theorem~\ref{thm:crash_opti}, one can show that, in the worst case, at most $\f{\Cluster_1}$ messages are sent by non-faulty replicas in $\Cluster_1$ and received by non-faulty replicas in $\Cluster_2$. Now consider the situation in which the faulty replicas in $\Cluster_1$ mimic the behavior in $C$ by sending certificates for another value $v'$ to the same receivers. For the replicas in $\Cluster_2$, the two runs behave the same, as in both cases at most $\f{\Cluster_1}$ certificates for a value, possibly signed by distinct replicas, are received. Hence, either both runs successfully send values, in which case $v'$ is received by $\Cluster_2$ without agreement, or both runs fail to send values. In both cases, \Name{P} fails to solve the cluster-sending problem.

Next, we prove the case for $\n{\Cluster_2} \geq \n{\Cluster_1}$ using contradiction. Assume that a protocol \Name{P} can solve the cluster-sending problem using at most $\tau_2 - 1$ certificates. Consider a run of \Name{P} that sends messages $C$, each message representing a single certificate, with $\abs{C} = \tau_2 - 1$. Let $R$ be the top $2\f{\Cluster_1}$ signers of certificates in $C$, let $C_R \subset C$ be the certificates signed by replicas in $R$, and let $D = C \difference C_R$. Via a contradiction argument similar to the one used in the proof of Theorem~\ref{thm:crash_opti}, one can show that $\abs{C_R} \geq 2q_2\f{\Cluster_1} + 2\f{\Cluster_1}\sgn{r}$ and $\abs{D} \leq q_2(\nf{\Cluster_1} - \f{\Cluster_1}) + r - 1 = \f{\Cluster_2}$. As $\abs{D} \leq \f{\Cluster_2}$, all replicas receiving these certificates could have crashed. Hence, the only certificates that are received by $\Cluster_2$ are in $C_R$. Partition $C_R$ into two sets of certificates $C_{R,1}$ and $C_{R,2}$ such that both sets contain certificates signed by at most $\f{\Cluster_1}$ distinct replicas. As the certificates in $C_{R,1}$ and $C_{R,2}$ are signed by $\f{\Cluster_1}$ distinct replicas, one of these sets can contain only certificates signed by Byzantine replicas. Hence, either $C_{R,1}$ or $C_{R,2}$ could certify a non-agreed upon value $v'$, while only the other set certifies $v$. Consequently, the replicas in $\Cluster_2$ cannot distinguish between receiving an agreed-upon value $v$ or a non-agreed-upon-value $v'$. We conclude that \Name{P} fails to solve the cluster-sending problem.
\end{proof}

\section{Cluster-sending via bijective sending}\label{sec:bijective}

In the previous section, we explored lower bounds for the cluster-sending problem. Close inspection shows that these lower bounds are only linear in the size of the clusters involved, which is much better than the quadratic bounds of straightforward broadcasting-based protocols. Hence, there is much room for improvement. Next, we develop \emph{bijective sending}, a powerful technique that allows the design of highly efficient cluster-sending protocols. In bijective sending, cluster $\Cluster_1$ uses the consensus protocol to reach agreement on a value $v$ and certificate $\Cert{v}{\Cluster_1}$. Then, the protocol chooses sets $S_1 \subseteq \Cluster_1$ and $S_2 \subseteq \Cluster_2$ of equal size and instruct each replica in $S_1 \subseteq \Cluster_1$ to send $(v, \Cert{v}{\Cluster_1})$ to a distinct replica in $\Cluster_2$. By choosing $S_1$ sufficiently large, we can guarantee successful cluster-sending. First, we present a bijective-sending protocol for systems with Byzantine failures and cluster signing.  The pseudo-code for this protocol, named \Name{BS-bcs}, can be found in Figure~\ref{fig:bs_bcs}. Next, we illustrate bijective sending, the underlying technique utilized by \Name{BS-bcs}:

\begin{figure}[t!]
    \begin{minipage}{0.7\textwidth}
        \begin{myprotocol}
            \TITLE{Protocol for the sending cluster $\Cluster_1$}
            \STATE Agree on $v$ and distribute $(v, \Cert{v}{\Cluster_1})$ to all non-faulty replicas in $\Cluster_1$.
            \STATE Choose replicas $S_1 \subseteq \Cluster_1$ with $\n{S_1} = \f{\Cluster_1} + \f{\Cluster_2} + 1$.
            \STATE Choose replicas $S_2 \subseteq \Cluster_2$ with $\n{S_2} = \f{\Cluster_1} + \f{\Cluster_2} + 1$.
            \STATE Choose a bijection $b : S_1 \rightarrow S_2$.
            \FOR{$\Replica_1 \in S_1$}
                \STATE $\Replica_1$ sends $(v, \Cert{v}{\Cluster_1})$ to $b(\Replica_1)$.\label{fig:bs_bcs:send}
            \ENDFOR
            \SPACE
            \TITLE{Protocol for the receiving cluster $\Cluster_2$}
            \STATE See the protocol for the receiving cluster in \Name{RB-bcs}.
        \end{myprotocol}
    \end{minipage}
    \hfill
    \begin{minipage}{0.285\textwidth}
        \caption{\Name{BS-bcs}, the bijective sending cluster-sending protocol that sends a value $v$ from $\Cluster_1$ to $\Cluster_2$. We assume Byzantine failures and a system that provides cluster signing.}\label{fig:bs_bcs}
    \end{minipage}
\end{figure}

\begin{example}\label{ex:bsent}
Let $\System$ be a system, let $\Cluster_1 = \{\Replica_1, \dots, \Replica_{8} \} \in \System$ with $\Faulty{\Cluster_1} = \{ \Replica_1, \Replica_3, \Replica_4 \}$, and let $\Cluster_2 = \{ \Replica_9, \dots, \Replica_{15} \} \in \System$ with $\Faulty{\Cluster_2} = \{ \Replica_9, \Replica_{11} \}$. We have $\f{\Cluster_1} + \f{\Cluster_2} + 1 = 6$. We choose
\[
    S_1 = \{ \Replica_2, \dots, \Replica_7 \};\qquad
    S_2 = \{ \Replica_9, \dots, \Replica_{15} \};\qquad
    b   = \{ \Replica_i \rightarrow \Replica_{i + 7} \mid 2 \leq i \leq 7 \}.
\]
In Figure~\ref{fig:ex_bs}, we sketched this situation. Replica $\Replica_2$ sends a valid message to $\Replica_9$. As $\Replica_9$ is faulty, it might ignore this message. Replicas $\Replica_3$ and $\Replica_4$ are faulty and might not send a valid message. Additionally, $\Replica_{11}$ is faulty and might ignore any message it receives. The messages sent from $\Replica_5$ to $\Replica_{12}$, from $\Replica_6$ to $\Replica_{13}$, and from $\Replica_7$ to $\Replica_{14}$ are all sent by non-faulty replicas to non-faulty replicas. Hence, these messages all arrive correctly.
\begin{figure}[t!]
    \begin{minipage}{0.5\textwidth}
        \centering
        \begin{tikzpicture}[scale=0.75,every node/.style={transform shape}]
            \node[left,node_text] at (-0.8, 0) {$\Cluster_1$:};
            \node[left,node_text] at (-0.8, -2) {$\Cluster_2$:};

            \draw[rounded corners,thick,fill=black!10] (-0.7, -0.5) rectangle (7.7, 0.7);
            \draw[rounded corners,thick,fill=black!10] ( 0.3, -2.7) rectangle (7.7, -1.5);
            \draw[rounded corners,draw=red!60,fill=red!60] (-0.6,  -0.4) rectangle (0.6, 0.6);
            \draw[rounded corners,draw=red!60,fill=red!60] ( 1.4,  -0.4) rectangle (3.6, 0.6);
            \draw[rounded corners,draw=red!60,fill=red!60] ( 0.4,  -2.6) rectangle (1.6, -1.6);
            \draw[rounded corners,draw=red!60,fill=red!60] ( 2.4,  -2.6) rectangle (3.6, -1.6);

            \node[dot] (r1)  at (0, 0) {};
            \node[dot] (r2)  at (1, 0) {};
            \node[dot] (r3)  at (2, 0) {};
            \node[dot] (r4)  at (3, 0) {};
            \node[dot] (r5)  at (4, 0) {};
            \node[dot] (r6)  at (5, 0) {};
            \node[dot] (r7)  at (6, 0) {};
            \node[dot] (r8)  at (7, 0) {};

            \node[dot] (r9)  at (1, -2) {};
            \node[dot] (r10) at (2, -2) {};
            \node[dot] (r11) at (3, -2) {};
            \node[dot] (r12) at (4, -2) {};
            \node[dot] (r13) at (5, -2) {};
            \node[dot] (r14) at (6, -2) {};
            \node[dot] (r15) at (7, -2) {};

            \node[above=-3pt,label] at (r1)  {$\Replica_{1}$};
            \node[above=-3pt,label] at (r2)  {$\Replica_{2}$};
            \node[above=-3pt,label] at (r3)  {$\Replica_{3}$};
            \node[above=-3pt,label] at (r4)  {$\Replica_{4}$};
            \node[above=-3pt,label] at (r5)  {$\Replica_{5}$};
            \node[above=-3pt,label] at (r6)  {$\Replica_{6}$};
            \node[above=-3pt,label] at (r7)  {$\Replica_{7}$};
            \node[above=-3pt,label] at (r8)  {$\Replica_{8}$};
            \node[below=-3pt,label] at (r9)  {$\Replica_{9}$};
            \node[below=-3pt,label] at (r10) {$\Replica_{10}$};
            \node[below=-3pt,label] at (r11) {$\Replica_{11}$};
            \node[below=-3pt,label] at (r12) {$\Replica_{12}$};
            \node[below=-3pt,label] at (r13) {$\Replica_{13}$};
            \node[below=-3pt,label] at (r14) {$\Replica_{14}$};
            \node[below=-3pt,label] at (r15) {$\Replica_{15}$};
            
            \path[thick] (r2) edge[dashed,->] (r9);
            \path[thick] (r3) edge[dashed,->] (r10);
            \path[thick] (r4) edge[dashed,->] (r11);
            \path[thick] (r5) edge[->] (r12);
            \path[thick] (r6) edge[->] (r13);
            \path[thick] (r7) edge[->] (r14);
        \end{tikzpicture}
    \end{minipage}
    \hfill
    \begin{minipage}{0.475\textwidth}
        \caption{Bijection sending from $\Cluster_1$ to $\Cluster_2$. The faulty replicas are highlighted using a red background. The edges connect replicas $\Replica \in \Cluster_1$ with $b(\Replica) \in \Cluster_2$. Each solid edge indicates a message sent and received by non-faulty replicas. Each dashed edge indicates a message sent or received by a faulty replica.}\label{fig:ex_bs}
    \end{minipage}
\end{figure}
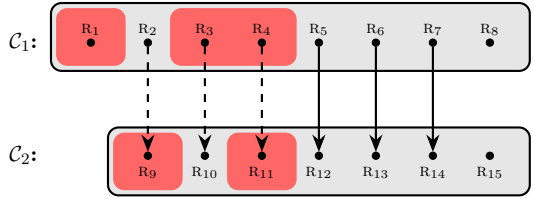
\end{example}

Having illustrated the concept of bijective sending, as employed by \Name{BS-bcs}, we are now ready to prove correctness of \Name{BS-bcs}:

\begin{proposition}\label{prop:bs_bcs}
Let $\System$ be a system with Byzantine failures and cluster signing and let $\Cluster_1, \Cluster_2 \in \System$. If $\n{\Cluster_1} > 2\f{\Cluster_1}$, $\n{\Cluster_1} > \f{\Cluster_1} + \f{\Cluster_2}$, and $\n{\Cluster_2} > \f{\Cluster_1} + \f{\Cluster_2}$, then \Name{BS-bcs} satisfies Definition~\ref{def:csp}. The protocol sends $\f{\Cluster_1} + \f{\Cluster_2} + 1$ messages, of size $\BigO{\Size{v}}$ each, between $\Cluster_1$ and $\Cluster_2$.
\end{proposition}
\begin{proof}
 Choose $S_1 \subseteq \Cluster_1$ and $S_2 \subseteq \Cluster_2$ in accordance with \Name{BS-bcs}. We have $\n{S_1} = \n{S_2} = \f{\Cluster_1} + \f{\Cluster_2} + 1$. Let $T = \{ b(\Replica) \mid \Replica \in \NonFaulty{S_1} \}$. By construction, we have $\nf{S_1} = \n{T} \geq \f{\Cluster_2} + 1$. Hence, we have $\nf{T} \geq 1$. Due to Line~\ref{fig:bs_bcs:send}, each replica in $\NonFaulty{T}$ will receive the message $(v, \Cert{v}{\Cluster_1})$ from a distinct replica in $\NonFaulty{S_1}$ and broadcast $(v, \Cert{v}{\Cluster_1})$ to all replicas in $\Cluster_2$. As $\nf{T} \geq 1$, each replica $\Replica_2' \in  \NonFaulty{\Cluster_2}$ will receive $(v, \Cert{v}{\Cluster_1})$ from a replica in $\Cluster_2$. Hence, analogous to the proof of Proposition~\ref{prop:rb_brs}, we can prove \emph{receipt}, \emph{confirmation}, and \emph{agreement}.
\end{proof}

As replica signing can emulate cluster signing, \Name{BS-bcs} can also be used for systems with only replica signing. Such an emulated solution does require large messages whose size depends on the size of the sending cluster, however. For systems with replica signing we can utilize bijective sending in another manner, however. We propose \Name{BS-brs}, for which the pseudo-code can be found in Figure~\ref{fig:bs_brs}. Next, we prove the correctness of \Name{BS-brs}:

\begin{figure}[t!]
    \begin{minipage}{0.7\textwidth}
        \begin{myprotocol}
        \TITLE{Protocol for the sending cluster $\Cluster_1$}
        \STATE Agree on $v$ and distribute $v$ to all non-faulty replicas in $\Cluster_1$.
        \STATE Choose replicas $S_1 \subseteq \Cluster_1$ with $\n{S_1} = 2\f{\Cluster_1} + \f{\Cluster_2} + 1$.
        \STATE Choose replicas $S_2 \subseteq \Cluster_2$ with $\n{S_2} = 2\f{\Cluster_1} + \f{\Cluster_2} + 1$.
        \STATE Choose a bijection $b : S_1 \rightarrow S_2$.
        \FOR{$\Replica_1 \in S_1$}
            \STATE $\Replica_1$ sends $(v, \Cert{v}{\Replica_1})$ to $b(\Replica_1)$.\label{fig:bs_brs:send}
        \ENDFOR
        \SPACE
        \TITLE{Protocol for the receiving cluster $\Cluster_2$}
        \STATE See the protocol for the receiving cluster in \Name{RB-brs}.
        \end{myprotocol}
    \end{minipage}
    \hfill
    \begin{minipage}{0.285\textwidth}
        \caption{\Name{BS-brs}, the bijective sending cluster-sending protocol that sends a value $v$ from $\Cluster_1$ to $\Cluster_2$. We assume Byzantine failures and a system that provides replica signing.}\label{fig:bs_brs}
    \end{minipage}
\end{figure}

\begin{proposition}\label{prop:bs_brs}
Let $\System$ be a system with Byzantine failures and replica signing and let $\Cluster_1, \Cluster_2 \in \System$. If  $\n{\Cluster_1} > 2\f{\Cluster_1} + \f{\Cluster_2}$ and $\n{\Cluster_2} > 2\f{\Cluster_1} + \f{\Cluster_2}$, then \Name{BS-brs} satisfies Definition~\ref{def:csp}. The protocol sends $2\f{\Cluster_1} + \f{\Cluster_2} + 1$ messages, of size $\BigO{\Size{v}}$ each, between $\Cluster_1$ and $\Cluster_2$.
\end{proposition}
\begin{proof}
Choose $S_1 \subseteq \Cluster_1$ and $S_2 \subseteq \Cluster_2$ in accordance with \Name{BS-brs}. We have $\n{S_1} = \n{S_2} =  2\f{\Cluster_1} + \f{\Cluster_2} + 1$. Let $T = \{ b(\Replica) \mid \Replica \in \NonFaulty{S_1} \}$. By construction, we have $\nf{S_1} = \n{T} \geq  \f{\Cluster_1} + \f{\Cluster_2} + 1$. Hence, we have $\nf{T} \geq \f{\Cluster_1} + 1$. Due to Line~\ref{fig:bs_brs:send}, each replica in $\NonFaulty{T}$ will receive the message $(v, \Cert{v}{\Replica_1})$ from a distinct replica $\Replica_1 \in \NonFaulty{S_1}$. Hence, analogous to the proof of Proposition~\ref{prop:rb_brs}, we can prove \emph{receipt}, \emph{confirmation}, and \emph{agreement}.
\end{proof}

For completeness, we consider the situation in which replica certificates have constant size. In this case, the presented version of \Name{BS-brs} performs too much communication.  We can correct this by only letting $\f{\Cluster_1} + \f{\Cluster_2} + 1$ replicas send the value $v$, while all $2\f{\Cluster_1} + \f{\Cluster_2} + 1$ replicas send a replica certificate.

\section{Cluster-sending via  partitioning}\label{sec:partitioned}

The bijective sending techniques introduced in the previous section have optimal communication complexity. Unfortunately, bijective sending is in practice limited to communication between similar-sized clusters, as it places unrealistic requirements on clusters that vastly differ in size. 
\begin{example}
Consider a system $\System$ with Byzantine failures and cluster certificates. The cluster $\Cluster_1 \in \System$ wants to send value $v$ to $\Cluster_2 \in \System$. Notice that \Name{BS-bcs} requires $\f{\Cluster_1} + \f{\Cluster_2} \leq \n{\Cluster_2}$. Hence, when using \Name{BS-bcs}, the number of faulty replicas in $\Cluster_1$ is upper-bounded by $\nf{\Cluster_2} \leq \n{\Cluster_2}$, this independent of the size of $\Cluster_1$.
\end{example}

Next, we show how to generalize bijective sending to arbitrary-sized clusters. We do so by \emph{partitioning} the larger-sized cluster into a set of smaller clusters, and then letting sufficient of these smaller clusters participate independent in bijective sending. First, we introduce the relevant partitioning notation.

\begin{definition}\label{def:cpart}
Let $\System$ be a system with $\Cluster  \in \System$, let $\mathcal{P}$ be a subset of the replicas in $\System$, let $c > 0$ be a constant, let $q = \n{\Cluster} \div c$, and let $r = \n{\Cluster} \bmod c$. A \emph{$c$-partition} $\Partition{\mathcal{P}} = \{ P_1, \dots, P_q, P'\}$ of $\mathcal{P}$ is a partition of the set of replicas $\mathcal{P}$ into sets $P_1, \dots, P_q, P'$ such that $\n{P_i} = c$, $1 \leq i \leq q$, and $\n{P'} = r$.
\end{definition}

\begin{example}\label{ex:cpart}
Let $\System$ be a system, let $\Cluster = \{\Replica_1, \dots, \Replica_{11} \} \in \System$, and let $\Faulty{\Cluster} = \{ \Replica_1, \dots, \Replica_5 \}$. The set $\Partition{\Cluster} = \{ P_1, P_2, P' \}$ with $P_1 = \{ \Replica_1, \dots, \Replica_4 \}$, $P_2 = \{ \Replica_5, \dots, \Replica_8 \}$, and $P' = \{\Replica_9, \Replica_{10}, \Replica_{11} \}$ is a $4$-partition of $\Cluster$. The cluster $\Cluster$ and the partition $\Partition{\Cluster}$ are illustrated in Figure~\ref{fig:ex_cpart}. We have $\Faulty{P_1} = P_1$, $\NonFaulty{P_1} = \emptyset$, and $\n{P_1} = \f{P_1} = 4$. Likewise, we have $\Faulty{P_2} = \{ \Replica_5 \}$, $\NonFaulty{P_2} = \{ \Replica_6, \Replica_7, \Replica_8 \}$, $\n{P_2} = 4$, and $\f{P_2} = 1$.

\begin{figure}[t!]
    \begin{minipage}{0.4\textwidth}
        \centering
        \begin{tikzpicture}[scale=0.8,every node/.style={transform shape}]
            \node[above,node_text] at (1.5, 2.7) {Cluster $\Cluster$:};
            \draw[rounded corners,thick,fill=black!10] (-1.2, -1) rectangle (4, 2.7);
            \draw[rounded corners,draw=red!60,fill=red!60] (-0.6,  -0.6) rectangle (3.6, 0.4);
            \draw[rounded corners,draw=red!60,fill=red!60] (-0.6,  -0.6) rectangle (0.6, 1.4);
            \draw[rounded corners,semithick,fill=blue!70,fill opacity=0.5] (-0.5, -0.5) rectangle (3.5, 0.3);
            \draw[rounded corners,semithick,fill=blue!70,fill opacity=0.5] (-0.5,  0.5) rectangle (3.5, 1.3);
            \draw[rounded corners,semithick,fill=blue!70,fill opacity=0.5] (-0.5,  1.5) rectangle (2.5, 2.3);

            \node[left,node_text] at (-0.5, -0.1) {$P_1$};
            \node[left,node_text] at (-0.5,  0.9) {$P_2$};
            \node[left,node_text] at (-0.5,  1.9) {$P'$};

            \node[dot] (r1)  at (0, 0) {};
            \node[dot] (r2)  at (1, 0) {};
            \node[dot] (r3)  at (2, 0) {};
            \node[dot] (r4)  at (3, 0) {};
            \node[dot] (r5)  at (0, 1) {};
            \node[dot] (r6)  at (1, 1) {};
            \node[dot] (r7)  at (2, 1) {};
            \node[dot] (r8)  at (3, 1) {};
            \node[dot] (r9)  at (0, 2) {};
            \node[dot] (r10) at (1, 2) {};
            \node[dot] (r11) at (2, 2) {};

            \node[below=-3pt,label] at (r1)  {$\Replica_{1}$};
            \node[below=-3pt,label] at (r2)  {$\Replica_{2}$};
            \node[below=-3pt,label] at (r3)  {$\Replica_{3}$};
            \node[below=-3pt,label] at (r4)  {$\Replica_{4}$};
            \node[below=-3pt,label] at (r5)  {$\Replica_{5}$};
            \node[below=-3pt,label] at (r6)  {$\Replica_{6}$};
            \node[below=-3pt,label] at (r7)  {$\Replica_{7}$};
            \node[below=-3pt,label] at (r8)  {$\Replica_{8}$};
            \node[below=-3pt,label] at (r9)  {$\Replica_{9}$};
            \node[below=-3pt,label] at (r10) {$\Replica_{10}$};
            \node[below=-3pt,label] at (r11) {$\Replica_{11}$};
        \end{tikzpicture}
    \end{minipage}
    \hfill
    \begin{minipage}{0.575\textwidth}
        \caption{An example of a $4$-partition of a cluster $\Cluster$ with $11$ replicas, of which the first five are faulty. The three partitions are grouped in blue boxes, the faulty replicas are highlighted using a red background.}\label{fig:ex_cpart}
    \end{minipage}
\end{figure}
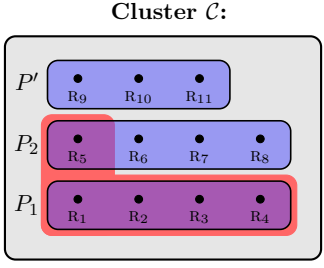
\end{example}

Having introduced partitioning, we are ready to generalize bijective sending to non-similar-sized clusters. Let $\Cluster_1$ be the sending cluster and $\Cluster_2$ be the receiving cluster. First, we consider the case $\n{\Cluster_2} \leq \f{\Cluster_1} + \f{\Cluster_2}$. The pseudo-code for the protocol, named \Name{SPBS-($\alpha$,$\zeta$)}, can be found in Figure~\ref{fig:spbs}. Next, we prove the correctness of specific instances of the protocol for Byzantine systems that provide either cluster signing or replica signing:
\begin{figure}[t!]
    \begin{minipage}{0.7\textwidth}
        \begin{myprotocol}
            \TITLE{Protocol for the sending cluster $\Cluster_1$}
            \STATE The agreement step of \Name{BS-$\zeta$} for value $v$.
            \STATE Choose replicas $\mathcal{P} \subseteq \Cluster_1$ with $\n{P} = \alpha$.
            \STATE Choose a $\n{\Cluster_2}$-partition $\Partition{\mathcal{P}}$ of $\mathcal{P}$.
            \FOR{$P \in \Partition{\mathcal{P}}$}
                \STATE Choose replicas $Q \subseteq \Cluster_2$ with $\n{Q} = \n{P}$.
                \STATE Choose a bijection $b : P \rightarrow Q$.
                \FOR{$\Replica_1 \in P$}
                    \STATE Send $v$ from $\Replica_1$ to $b(\Replica_1)$ via the send step of \Name{BS-$\zeta$}.
                \ENDFOR
            \ENDFOR
            \SPACE
            \TITLE{Protocol for the receiving cluster $\Cluster_2$}
            \STATE See the protocol for the receiving cluster in \Name{BS-$\zeta$}.
        \end{myprotocol}
    \end{minipage}
    \hfill
    \begin{minipage}{0.285\textwidth}
        \caption{\Name{SPBS-($\alpha$,$\zeta$)}, $\zeta \in \{ \Name{bcs}, \Name{brs} \}$, the sender-partitioned bijective sending cluster-sending protocol that sends a value $v$ from $\Cluster_1$ to $\Cluster_2$. We assume the same system properties as \Name{BS-$\zeta$}.}\label{fig:spbs}
    \end{minipage}
\end{figure}

\begin{proposition}\label{prop:spbs}
Let $\System$ be a system with Byzantine failures and let $\Cluster_1, \Cluster_2 \in \System$, $\sigma_1$ as defined in Theorem~\ref{thm:crash_opti}, and $\tau_1$ as defined in Theorem~\ref{thm:cert_opti}.
\begin{enumerate}[nosep]
    \item If $\System$ provides cluster signing and $\sigma_1 \leq \n{\Cluster_1}$, then \Name{SPBS-($\sigma_1$,bcs)} satisfies Definition~\ref{def:csp}. The protocol sends $\sigma_1$ messages, of size $\BigO{\Size{v}}$ each, between $\Cluster_1$ and $\Cluster_2$.
    \item If $\System$ provides replica signing and $\tau_1 \leq \n{\Cluster_1}$, then \Name{SPBS-($\tau_1$,brs)} satisfies Definition~\ref{def:csp}. The protocol sends $\tau_1$ messages, of size $\BigO{\Size{v}}$ each, between $\Cluster_1$ and $\Cluster_2$.
\end{enumerate}
\end{proposition}
\begin{proof}
Let $\beta = (\f{\Cluster_1} + 1)$ in the case of cluster signing and let $\beta = (2\f{\Cluster_1} + 1)$ in the case of replica signing. Let $q = \beta \div \nf{\Cluster_2}$ and $r = \beta \bmod \nf{\Cluster_2}$. We have $\alpha = q\n{\Cluster_2} + r + \f{\Cluster_2}\sgn{r}$. Choose $\mathcal{P}$ and choose $\Partition{\mathcal{P}} = \{ P_1, \dots, P_{q}, P' \}$ in accordance with \Name{SPBS-($\alpha$,$\zeta$)}. For each $P \in \mathcal{P}$, choose a $Q$ and $b$ in accordance with \Name{SPBS-($\alpha$,$\zeta$)}, and let $z(P) = \{ \Replica \in P \mid b(\Replica) \in \Faulty{Q} \}$. As each such $b$ has a distinct domain, the union of them is a surjection $f : \mathcal{P} \rightarrow \n{\Cluster_2}$. By construction, we have $\n{P'} =  r + \f{\Cluster_2}\sgn{r}$, $\n{z(P')} \leq \f{\Cluster_2}\sgn{r}$, and, for every $i$, $1 \leq i \leq q$, $\n{P_i} = \n{\Cluster_2}$ and $\n{z(P_i)} = \f{\Cluster_2}$. Let $V = \mathcal{P} \difference \bigl(\bigcup_{P \in \Partition{\mathcal{P}}}\ z(P)\bigr)$. We have 
\[
    \n{V} \geq \n{\mathcal{P}} - (q\f{\Cluster_2} + \f{\Cluster_2}\sgn{r})
          = (q\n{\Cluster_2} + r + \f{\Cluster_2}\sgn{r}) - (q\f{\Cluster_2} + \f{\Cluster_2}\sgn{r})
          = q\nf{\Cluster_2} + r = \beta.
\]
Let $T = \{ f(\Replica) \mid \Replica \in \NonFaulty{V} \}$.  By construction, we have $\nf{T} = \n{T}$. To complete the proof, we consider cluster signing and replica signing separately. First, the case for cluster signing. As $\n{V} \geq \beta = \f{\Cluster_1} + 1$, we have $\nf{V} \geq 1$. By construction, the replicas in $\NonFaulty{T}$ will receive the messages $(v, \Cert{v}{\Cluster_1})$ from the replicas $\Replica_1 \in V$. Hence, analogous to the proof of Proposition~\ref{prop:rb_bcs}, we can prove \emph{receipt}, \emph{confirmation}, and \emph{agreement}. Finally, the case for replica signing. As $\n{V} \geq \beta = 2\f{\Cluster_1} + 1$, we have $\nf{V} \geq \f{\Cluster_1} + 1$.  By construction, the replicas in $\NonFaulty{T}$ will receive the messages $(v, \Cert{v}{\Replica_1})$ from each replica $\Replica_1 \in V$. Hence, analogous to the proof of Proposition~\ref{prop:rb_brs}, we can prove \emph{receipt}, \emph{confirmation}, and \emph{agreement}.
\end{proof}

Next, we consider the case $\n{\Cluster_1} \leq \f{\Cluster_1} + \f{\Cluster_2}$. The pseudo-code for the protocol, named \Name{RPBS-($\alpha$,$\zeta$)}, can be found in Figure~\ref{fig:rpbs}. Next, we prove the correctness of specific instances of the protocol for Byzantine systems that provide either cluster signing or replica signing:

\begin{figure}[t!]
    \begin{minipage}{0.7\textwidth}
        \begin{myprotocol}
            \TITLE{Protocol for the sending cluster $\Cluster_1$}
            \STATE The agreement step of \Name{BS-$\zeta$} for value $v$.
            \STATE Choose replicas $\mathcal{P} \subseteq \Cluster_2$ with $\n{P} = \alpha$.
            \STATE Choose a $\n{\Cluster_1}$-partition $\Partition{\mathcal{P}}$ of $\mathcal{P}$.
            \FOR{$P \in \Partition{\mathcal{P}}$}
                \STATE Choose replicas $Q \subseteq \Cluster_1$ with $\n{Q} = \n{P}$.
                \STATE Choose a bijection $b : Q \rightarrow P$.
                \FOR{$\Replica_1 \in Q$}
                    \STATE Send $v$ from $\Replica_1$ to $b(\Replica_1)$ via the send step of \Name{BS-$\zeta$}.
                \ENDFOR
            \ENDFOR
            \SPACE
            \TITLE{Protocol for the receiving cluster $\Cluster_2$}
            \STATE See the protocol for the receiving cluster in \Name{BS-$\zeta$}.
        \end{myprotocol}
    \end{minipage}
    \hfill
    \begin{minipage}{0.285\textwidth}
        \caption{\Name{RPBS-($\alpha$,$\zeta$)}, $\zeta \in \{ \Name{bcs}, \Name{brs} \}$, the receiver-partitioned bijective sending cluster-sending protocol that sends a value $v$ from $\Cluster_1$ to $\Cluster_2$. We assume the same system properties as \Name{BS-$\zeta$}.}\label{fig:rpbs}
    \end{minipage}
\end{figure}

\begin{proposition}\label{prop:rpbs}
Let $\System$ be a system with Byzantine failures and let $\Cluster_1, \Cluster_2 \in \System$, $\sigma_2$ as defined in Theorem~\ref{thm:crash_opti}, and $\tau_2$ as defined in Theorem~\ref{thm:cert_opti}.
\begin{enumerate}[nosep]
    \item If $\System$ provides cluster signing and $\sigma_2 \leq \n{\Cluster_2}$, then \Name{RPBS-($\sigma_2$,bcs)} satisfies Definition~\ref{def:csp}. The protocol sends $\sigma_2$ messages, of size $\BigO{\Size{v}}$ each, between $\Cluster_1$ and $\Cluster_2$.
    \item If $\System$ provides replica signing and $\tau_2 \leq \n{\Cluster_2}$, then \Name{RPBS-($\tau_2$,brs)} satisfies Definition~\ref{def:csp}. The protocol sends $\tau_2$ messages, of size $\BigO{\Size{v}}$ each, between $\Cluster_1$ and $\Cluster_2$.
\end{enumerate}
\end{proposition}
\begin{proof}
Let $\beta = \nf{\Cluster_1}$ and $\gamma = 1$ in the case of cluster signing and let $\beta = (\nf{\Cluster_1} - \f{\Cluster_1})$ and $\gamma = 2$ in the case of replica signing. Let $q = (\f{\Cluster_2} + 1) \div \beta$ and $r = (\f{\Cluster_2} + 1) \bmod \beta$. We have $\alpha =  q\n{\Cluster_1} + r + \gamma\f{\Cluster_1}\sgn{r}$. Choose $\mathcal{P}$ and choose $\Partition{\mathcal{P}} = \{ P_1, \dots, P_{q}, P' \}$ in accordance with \Name{RPBS-($\alpha$,$\zeta$)}. For each $P \in \mathcal{P}$, choose a $Q$ and $b$ in accordance with \Name{RPBS-($\alpha$,$\zeta$)}, and let $z(P) = \{ \Replica \in P \mid b^{-1}(\Replica) \in \Faulty{Q} \}$. As each such $b^{-1}$ has a distinct domain, the union of them is a surjection $f^{-1} : \mathcal{P} \rightarrow \n{\Cluster_1}$. By construction, we have $\n{P'} =  r + \gamma\f{\Cluster_1}\sgn{r}$, $\n{z(P')} \leq \f{\Cluster_1}\sgn{r}$, and, for every $i$, $1 \leq i \leq q$, $\n{P_i} = \n{\Cluster_1}$ and $\n{z(P_i)} = \f{\Cluster_1}$. Let $T = \mathcal{P} \difference \big(\bigcup_{P \in \Partition{\mathcal{P}}}\ z(P)\big)$. We have 
\[ \n{T} \geq \n{\mathcal{P}} - (q\f{\Cluster_1} + \f{\Cluster_1}\sgn{r})
          = (q\n{\Cluster_1} + r + \gamma\f{\Cluster_1}\sgn{r}) - (q\f{\Cluster_1} + \f{\Cluster_1}\sgn{r})
          = q\nf{\Cluster_1} + r + (\gamma - 1)\f{\Cluster_1}\sgn{r}.
\]
To complete the proof, we consider cluster signing and replica signing separately.

First, the case for cluster signing. We have $\beta = \nf{\Cluster_1}$ and $\gamma = 1$. Hence,
    \[  \n{T} \geq q\nf{\Cluster_1} + r + (\gamma - 1)\f{\Cluster_1}\sgn{r}
                = q\beta + r = \f{\Cluster_2} + 1.
    \]
We have $\nf{T} \geq \n{T} - \f{\Cluster_2} \geq 1$. Let $V = \{ f^{-1}(\Replica) \mid \Replica \in \NonFaulty{T} \}$. By construction, we have $\nf{V} = \n{V}$ and we have $\nf{V} \geq 1$. Consequently, the replicas in $\NonFaulty{T}$ will receive the messages $(v, \Cert{v}{\Cluster_1})$ from the replicas $\Replica_1 \in V$. Analogous to the proof of Proposition~\ref{prop:rb_bcs}, we can prove \emph{receipt}, \emph{confirmation}, and \emph{agreement}.
        
Finally, the case for replica signing. We have $\beta = \nf{\Cluster_1} - \f{\Cluster_1}$ and $\gamma = 2$. Hence,
\begin{multline*}
            \n{T} \geq q\nf{\Cluster_1} + r + (\gamma - 1)\f{\Cluster_1}\sgn{r}
                  = q(\beta + \f{\Cluster_1}) + r + \f{\Cluster_1}\sgn{r}\\{}
                  = (q\beta + r) + q\f{\Cluster_1} + \f{\Cluster_1}\sgn{r}
                  = (\f{\Cluster_2} + 1) + q\f{\Cluster_1} + \f{\Cluster_1}\sgn{r}.
\end{multline*}
We have $\nf{T} \geq q\f{\Cluster_1} + \f{\Cluster_1}\sgn{r} + 1 = (q + \sgn{r})\f{\Cluster_1} + 1$. As there are $(q + \sgn{r})$ non-empty sets in $\Partition{\mathcal{P}}$, there must be a set $P \in \mathcal{P}$ with $\n{P \intersect \nf{T}} \geq \f{\Cluster_1} + 1$. Let $b$ be the bijection chosen earlier for $P$ and let $V = \{ b^{-1}(\Replica) \mid \Replica \in (P \intersect \nf{T}) \}$. By construction, we have $\nf{V} = \n{V}$ and we have $\nf{V} \geq \f{\Cluster_1} + 1$. Consequently, the replicas in $\NonFaulty{T}$ will receive the messages $(v, \Cert{v}{\Replica_1})$ from each replica $\Replica_1 \in V$. Hence, analogous to the proof of Proposition~\ref{prop:rb_brs}, we can prove \emph{receipt}, \emph{confirmation}, and \emph{agreement}.
\end{proof}

As with \Name{BS-brs}, \Name{SPBS-($\tau_1$,brs)}  and \Name{RPBS-($\tau_2$,brs)} can be optimized for the case in which replica certificates have constant size. In these cases, we only let $\sigma_1$ or $\sigma_2$ replicas send the value $v$, respectively, while all $\tau_1$ and $\tau_2$ replicas send a replica certificate, respectively. 

The bijective sending cluster-sending protocols, the sender-partitioned bijective cluster-sending protocols, and the receiver-partitioned bijective cluster-sending protocols each deal with differently-sized clusters. Furthermore, we can use the protocols designed with cluster certificates in mind also in the other cases using the cluster certificate emulation strategies discussed in Section~\ref{sec:formal}. By choosing the applicable protocols, we have the following:

\begin{theorem}\label{thm:opti_csp}
Let $\System$ be a system and let $\Cluster_1, \Cluster_2 \in \System$. Consider the cluster-sending problem in which $\Cluster_1$ sends a value $v$ to $\Cluster_2$.
\begin{enumerate}[nosep]
    \item If $\n{\Cluster} > 3\f{\Cluster}$, $\Cluster \in \System$, and $\System$ has crash failures, omit failures, or Byzantine failures and cluster signing, then \Name{BS-bcs}, \Name{SPBS-($\sigma_1$,bcs)}, and \Name{RPBS-($\sigma_2$,bcs)} are a solution to the cluster-sending problem with optimal message complexity. These protocols solve the cluster-sending problem using $\BigO{\max(\n{\Cluster_1}, \n{\Cluster_2})}$ messages, of size $\BigO{\Size{v}}$ each.
    \item If $\n{\Cluster} > 4\f{\Cluster}$, $\Cluster \in \System$, and $\System$ has Byzantine failures and replica sending,  then \Name{BS-brs}, \Name{SPBS-($\tau_1$,brs)}, and \Name{RPBS-($\tau_2$,brs)} are a solution to the cluster-sending problem  with optimal replica certificate usage. These protocols solve the cluster-sending problem using $\BigO{\max(\n{\Cluster_1}, \n{\Cluster_2})}$ messages, of size $\BigO{\Size{v}}$ each.
    \item If $\n{\Cluster} > 3\f{\Cluster}$, $\Cluster \in \System$, and $\System$ has Byzantine failures and replica sending, then \Name{BS-bcs}, \Name{SPBS-($\sigma_1$,bcs)}, and \Name{RPBS-($\sigma_2$,bcs)} are a solution to the cluster-sending problem. These protocols solve the cluster-sending problem using $\BigO{\max(\n{\Cluster_1}, \n{\Cluster_2})}$ messages, of size $\BigO{\Size{v} + \f{\Cluster_1}}$ each.
\end{enumerate}
\end{theorem}

\section{Conclusions and discussion}\label{sec:conclude}

In this paper, we formalized the \emph{cluster-sending problem}, the problem of sending messages between clusters that can have faulty replicas. We proved fundamental lower bounds on the complexity of the cluster-sending problem. We also developed two powerful techniques, bijective sending and partitioned bijective sending, that can be used in the construction of practical cluster-sending protocols with optimal complexity. Our work provides a strong foundation for the further exploration of novel designs that address challenges encountered in resilient distributed systems. These fundamental results open a number of key research avenues. 

First, the optimal protocols we propose apply to most practical situations, but in some extreme cases only the straightforward broadcasting-based protocols are applicable. It remains open whether we can improve on these broadcast-based protocols in all cases. Second, based on the assumptions made in this paper, we also foresee three fundamental opportunities for further study and development:
\begin{enumerate}[label=(\arabic*)]
\item The presence of \emph{public-key cryptography} (replica signing or cluster signing). Without these tools, each replica can only reliable detect the sender of messages it receives from other non-faulty replicas and forwarding messages becomes much harder. Hence, we can only imagine a significant increase in the complexity of the cluster-sending problem.
\item We operate in a fully \emph{dynamic failure model} in which the set of faulty replicas is ever changing. The leader-less protocols we designed operate perfectly under this restriction. In many practical settings the set of faulty replicas is relatively stable, however. It remains open to what degree cluster-sending can be optimized to such an optimistic assumption about failures to reduce the \emph{expected} complexity. As an example, we mention the usage of a dedicated reliable \emph{leader} responsible for coordinating incoming and outgoing communication. Such a design, with all its challenges, has already seen limited usage in scalable BFT systems such as Steward~\cite{steward}.
\item Going beyond \emph{reliable networks}. Assuming that the network is reliable enabled us to design one-way protocols without any message acknowledgement phases. Consequently, the protocols we present leverage network reliability to provide \emph{confirmation}. Alternatively, our protocols can be extended to provide a best-case effort to detect and recover from network unreliability (as far as possible~\cite{cap12,capproof,capthm}), which necessitates communication in both directions and will affect the lower bounds on the complexity of cluster-sending. 
\end{enumerate}

\bibliographystyle{plainurl}
\bibliography{sources}

\begin{thebibliography}{10}

\bibitem{steward}
Yair Amir, Claudiu Danilov, Danny Dolev, Jonathan Kirsch, John Lane, Cristina
  Nita-Rotaru, Josh Olsen, and David Zage.
\newblock Steward: Scaling byzantine fault-tolerant replication to wide area
  networks.
\newblock {\em IEEE Transactions on Dependable and Secure Computing},
  7(1):80--93, 2010.
\newblock \href {http://dx.doi.org/10.1109/TDSC.2008.53}
  {\path{doi:10.1109/TDSC.2008.53}}.

\bibitem{blockdev}
GSM Association.
\newblock Blockchain for development: Emerging opportunities for mobile,
  identity and aid, 2017.
\newblock URL:
  \url{https://www.gsma.com/mobilefordevelopment/wp-content/uploads/2017/12/Blockchain-for-Development.pdf}.

\bibitem{bft700}
Pierre-Louis Aublin, Rachid Guerraoui, Nikola Kne\v{z}evi\'{c}, Vivien
  Qu{\'e}ma, and Marko Vukoli\'{c}.
\newblock The next 700 bft protocols.
\newblock {\em ACM Transactions on Computer Systems}, 32(4):12:1--12:45, 2015.
\newblock \href {http://dx.doi.org/10.1145/2658994}
  {\path{doi:10.1145/2658994}}.

\bibitem{rbft}
Pierre-Louis Aublin, Sonia~Ben Mokhtar, and Vivien Qu{\'e}ma.
\newblock {RBFT}: Redundant byzantine fault tolerance.
\newblock In {\em 2013 IEEE 33rd International Conference on Distributed
  Computing Systems}, pages 297--306. IEEE, 2013.
\newblock \href {http://dx.doi.org/10.1109/ICDCS.2013.53}
  {\path{doi:10.1109/ICDCS.2013.53}}.

\bibitem{scaling}
Christian Berger and Hans~P. Reiser.
\newblock Scaling byzantine consensus: A broad analysis.
\newblock In {\em Proceedings of the 2Nd Workshop on Scalable and Resilient
  Infrastructures for Distributed Ledgers}, SERIAL'18, pages 13--18. ACM, 2018.
\newblock \href {http://dx.doi.org/10.1145/3284764.3284767}
  {\path{doi:10.1145/3284764.3284767}}.

\bibitem{blockeu}
Burkhard Blechschmidt.
\newblock Blockchain in {Europe}: Closing the strategy gap.
\newblock Technical report, Cognizant Consulting, 2018.
\newblock URL:
  \url{https://www.cognizant.com/whitepapers/blockchain-in-europe-closing-the-strategy-gap-codex3320.pdf}.

\bibitem{cap12}
Eric Brewer.
\newblock {CAP} twelve years later: How the ``rules'' have changed.
\newblock {\em Computer}, 45(2):23--29, 2012.
\newblock \href {http://dx.doi.org/10.1109/MC.2012.37}
  {\path{doi:10.1109/MC.2012.37}}.

\bibitem{capthm}
Eric~A. Brewer.
\newblock Towards robust distributed systems (abstract).
\newblock In {\em Proceedings of the Nineteenth Annual ACM Symposium on
  Principles of Distributed Computing}, pages 7--7. ACM, 2000.
\newblock \href {http://dx.doi.org/10.1145/343477.343502}
  {\path{doi:10.1145/343477.343502}}.

\bibitem{impactblock}
Michael Casey, Jonah Crane, Gary Gensler, Simon Johnson, and Neha Narula.
\newblock The impact of blockchain technology on finance: A catalyst for
  change.
\newblock Technical report, International Center for Monetary and Banking
  Studies, 2018.
\newblock URL:
  \url{https://www.cimb.ch/uploads/1/1/5/4/115414161/geneva21_1.pdf}.

\bibitem{pbft}
Miguel Castro and Barbara Liskov.
\newblock Practical byzantine fault tolerance.
\newblock In {\em Proceedings of the Third Symposium on Operating Systems
  Design and Implementation}, pages 173--186. USENIX Association, 1999.

\bibitem{pbftj}
Miguel Castro and Barbara Liskov.
\newblock Practical byzantine fault tolerance and proactive recovery.
\newblock {\em ACM Transactions on Computer Systems}, 20(4):398--461, 2002.
\newblock \href {http://dx.doi.org/10.1145/571637.571640}
  {\path{doi:10.1145/571637.571640}}.

\bibitem{christies}
Christie's.
\newblock Major collection of the fall auction season to be recorded with
  blockchain technology, 2018.
\newblock URL:
  \url{https://www.christies.com/presscenter/pdf/9160/RELEASE_ChristiesxArtoryxEbsworth_9160_1.pdf}.

\bibitem{aardvark}
Allen Clement, Edmund Wong, Lorenzo Alvisi, Mike Dahlin, and Mirco Marchetti.
\newblock Making byzantine fault tolerant systems tolerate byzantine faults.
\newblock In {\em Proceedings of the 6th USENIX Symposium on Networked Systems
  Design and Implementation}, pages 153--168. USENIX Association, 2009.

\bibitem{ibmgdpr}
Cindy Compert, Maurizio Luinetti, and Bertrand Portier.
\newblock Blockchain and {GDPR}: How blockchain could address five areas
  associated with gdpr compliance.
\newblock Technical report, IBM Security, 2018.
\newblock URL:
  \url{https://public.dhe.ibm.com/common/ssi/ecm/61/en/61014461usen/security-ibm-security-solutions-wg-white-paper-external-61014461usen-20180319.pdf}.

\bibitem{dolevstrong}
D.~Dolev and H.~Strong.
\newblock Authenticated algorithms for byzantine agreement.
\newblock {\em SIAM Journal on Computing}, 12(4):656--666, 1983.
\newblock \href {http://dx.doi.org/10.1137/0212045}
  {\path{doi:10.1137/0212045}}.

\bibitem{wurblockfood}
Lan Ge, Christopher Brewster, Jacco Spek, Anton Smeenk, and Jan Top.
\newblock Blockchain for agriculture and food: Findings from the pilot study.
\newblock Technical report, Wageningen University, 2017.
\newblock URL:
  \url{https://www.wur.nl/nl/Publicatie-details.htm?publicationId=publication-way-353330323634}.

\bibitem{algorand}
Yossi Gilad, Rotem Hemo, Silvio Micali, Georgios Vlachos, and Nickolai
  Zeldovich.
\newblock Algorand: Scaling byzantine agreements for cryptocurrencies.
\newblock In {\em Proceedings of the 26th Symposium on Operating Systems
  Principles}, pages 51--68. ACM, 2017.
\newblock \href {http://dx.doi.org/10.1145/3132747.3132757}
  {\path{doi:10.1145/3132747.3132757}}.

\bibitem{capproof}
Seth Gilbert and Nancy Lynch.
\newblock Brewer's conjecture and the feasibility of consistent, available,
  partition-tolerant web services.
\newblock {\em SIGACT News}, 33(2):51--59, 2002.
\newblock \href {http://dx.doi.org/10.1145/564585.564601}
  {\path{doi:10.1145/564585.564601}}.

\bibitem{blockhealthfac}
William~J. Gordon and Christian Catalini.
\newblock Blockchain technology for healthcare: Facilitating the transition to
  patient-driven interoperability.
\newblock {\em Computational and Structural Biotechnology Journal},
  16:224--230, 2018.
\newblock \href {http://dx.doi.org/10.1016/j.csbj.2018.06.003}
  {\path{doi:10.1016/j.csbj.2018.06.003}}.

\bibitem{mbftba}
Suyash Gupta, Jelle Hellings, and Mohammad Sadoghi.
\newblock Brief announcement: revisiting consensus protocols through wait-free
  parallelization.
\newblock In {\em 33nd International Symposium on Distributed Computing}, 2019.

\bibitem{encybd}
Suyash Gupta and Mohammad Sadoghi.
\newblock {\em Blockchain Transaction Processing}, pages 1--11.
\newblock Springer International Publishing, 2018.
\newblock \href {http://dx.doi.org/10.1007/978-3-319-63962-8_333-1}
  {\path{doi:10.1007/978-3-319-63962-8_333-1}}.

\bibitem{clussendba}
Jelle Hellings and Mohammad Sadoghi.
\newblock Brief announcement: the fault-tolerant cluster-sending problem.
\newblock In {\em 33nd International Symposium on Distributed Computing}, 2019.

\bibitem{promiseblock}
Matt Higginson, Johannes-Tobias Lorenz, Björn Münstermann, and Peter~Braad
  Olesen.
\newblock The promise of blockchain.
\newblock Technical report, McKinsey\&Company, 2017.
\newblock URL:
  \url{https://www.mckinsey.com/industries/financial-services/our-insights/the-promise-of-blockchain}.

\bibitem{blockhealthover}
Maged~N. Kamel~Boulos, James~T. Wilson, and Kevin~A. Clauson.
\newblock Geospatial blockchain: promises, challenges, and scenarios in health
  and healthcare.
\newblock {\em International Journal of Health Geographics}, 17(1):1211--1220,
  2018.
\newblock \href {http://dx.doi.org/10.1186/s12942-018-0144-x}
  {\path{doi:10.1186/s12942-018-0144-x}}.

\bibitem{cheapbft}
R\"{u}diger Kapitza, Johannes Behl, Christian Cachin, Tobias Distler, Simon
  Kuhnle, Seyed~Vahid Mohammadi, Wolfgang Schr\"{o}der-Preikschat, and Klaus
  Stengel.
\newblock {CheapBFT}: Resource-efficient byzantine fault tolerance.
\newblock In {\em Proceedings of the 7th ACM European Conference on Computer
  Systems}, pages 295--308. ACM, 2012.
\newblock \href {http://dx.doi.org/10.1145/2168836.2168866}
  {\path{doi:10.1145/2168836.2168866}}.

\bibitem{zyzzyva}
Ramakrishna Kotla, Lorenzo Alvisi, Mike Dahlin, Allen Clement, and Edmund Wong.
\newblock {Zyzzyva}: Speculative byzantine fault tolerance.
\newblock In {\em Proceedings of Twenty-first ACM SIGOPS Symposium on Operating
  Systems Principles}, pages 45--58. ACM, 2007.
\newblock \href {http://dx.doi.org/10.1145/1294261.1294267}
  {\path{doi:10.1145/1294261.1294267}}.

\bibitem{zyzzyvaj}
Ramakrishna Kotla, Lorenzo Alvisi, Mike Dahlin, Allen Clement, and Edmund Wong.
\newblock {Zyzzyva}: Speculative byzantine fault tolerance.
\newblock {\em ACM Transactions on Computer Systems}, 27(4):7:1--7:39, 2009.
\newblock \href {http://dx.doi.org/10.1145/1658357.1658358}
  {\path{doi:10.1145/1658357.1658358}}.

\bibitem{paxos}
Leslie Lamport.
\newblock The implementation of reliable distributed multiprocess systems.
\newblock {\em Computer Networks (1976)}, 2(2):95--114, 1978.
\newblock \href {http://dx.doi.org/10.1016/0376-5075(78)90045-4}
  {\path{doi:10.1016/0376-5075(78)90045-4}}.

\bibitem{paxossimple}
Leslie Lamport.
\newblock Paxos made simple.
\newblock {\em ACM SIGACT News, Distributed Computing Column 5}, 32(4):51--58,
  2001.
\newblock \href {http://dx.doi.org/10.1145/568425.568433}
  {\path{doi:10.1145/568425.568433}}.

\bibitem{fastbft}
Jian Liu, Wenting Li, Ghassan~O. Karame, and N.~Asokan.
\newblock Scalable byzantine consensus via hardware-assisted secret sharing.
\newblock {\em IEEE Transactions on Computers}, 68(1):139--151, 2019.
\newblock \href {http://dx.doi.org/10.1109/TC.2018.2860009}
  {\path{doi:10.1109/TC.2018.2860009}}.

\bibitem{hac}
Alfred~J. Menezes, Scott~A. Vanstone, and Paul C.~Van Oorschot.
\newblock {\em Handbook of Applied Cryptography}.
\newblock CRC Press, Inc., 1st edition, 1996.

\bibitem{bitcoin}
Satoshi Nakamoto.
\newblock Bitcoin: A peer-to-peer electronic cash system.
\newblock URL: \url{https://bitcoin.org/en/bitcoin-paper}.

\bibitem{blockplane}
Faisal Nawab and Mohammad Sadoghi.
\newblock Blockplane: A global-scale byzantizing middleware.
\newblock In {\em 35th International Conference on Data Engineering (ICDE)},
  pages 124--135. IEEE, 2019.
\newblock \href {http://dx.doi.org/10.1109/ICDE.2019.00020}
  {\path{doi:10.1109/ICDE.2019.00020}}.

\bibitem{dbdist}
M.~Tamer {\"O}zsu and Patrick Valduriez.
\newblock {\em Principles of Distributed Database Systems}.
\newblock Springer New York, 3th edition, 2011.

\bibitem{consbound}
M.~Pease, R.~Shostak, and L.~Lamport.
\newblock Reaching agreement in the presence of faults.
\newblock {\em Journal of the ACM}, 27(2):228--234, 1980.
\newblock \href {http://dx.doi.org/10.1145/322186.322188}
  {\path{doi:10.1145/322186.322188}}.

\bibitem{hypereal}
Michael Pisa and Matt Juden.
\newblock Blockchain and economic development: Hype vs. reality.
\newblock Technical report, Center for Global Development, 2017.
\newblock URL:
  \url{https://www.cgdev.org/publication/blockchain-and-economic-development-hype-vs-reality}.

\bibitem{pwcenergy}
PwC.
\newblock Blockchain -- an opportunity for energy producers and consumers?,
  2016.
\newblock URL:
  \url{https://www.pwc.com/gx/en/industries/energy-utilities-resources/publications/opportunity-for-energy-producers.html}.

\bibitem{rsasign}
Victor Shoup.
\newblock Practical threshold signatures.
\newblock In {\em Advances in Cryptology --- EUROCRYPT 2000}, pages 207--220.
  Springer Berlin Heidelberg, 2000.
\newblock \href {http://dx.doi.org/10.1007/3-540-45539-6_15}
  {\path{doi:10.1007/3-540-45539-6_15}}.

\bibitem{distalgo}
Gerard Tel.
\newblock {\em Introduction to Distributed Algorithms}.
\newblock Cambridge University Press, 2nd edition, 2001.

\bibitem{distsys}
Maarten van Steen and Andrew~S. Tanenbaum.
\newblock {\em Distributed Systems}.
\newblock Maarten van Steen, 3th edition, 2017.
\newblock URL: \url{https://www.distributed-systems.net/}.

\bibitem{minbft}
Santos~Veronese Veronese, Miguel Correia, Alysson~Neves Bessani, Lau~Cheuk
  Lung, and Paulo Verissimo.
\newblock Efficient byzantine fault-tolerance.
\newblock {\em IEEE Transactions on Computers}, 62(1):16--30, 2013.
\newblock \href {http://dx.doi.org/10.1109/TC.2011.221}
  {\path{doi:10.1109/TC.2011.221}}.

\bibitem{ethereum}
Gavin Wood.
\newblock Ethereum: a secure decentralised generalised transaction ledger.
\newblock {EIP}-150 revision.
\newblock URL: \url{https://gavwood.com/paper.pdf}.

\end{thebibliography}

\end{document}